\documentclass[11pt,reqno]{amsart}
\usepackage[english]{babel}
\usepackage{mathrsfs}
\usepackage{xcolor}
\usepackage[pdfencoding=auto, psdextra]{hyperref}
\hypersetup{colorlinks=true,allcolors=[rgb]{0,0,0.6}}
\usepackage[hmargin=2.7cm,vmargin=2.6cm]{geometry}

\newtheorem{theorem}{Theorem}[section]
\newtheorem{proposition}[theorem]{Proposition}
\newtheorem{lemma}[theorem]{Lemma}
\newtheorem{cor}[theorem]{Corollary}
\theoremstyle{definition}
\newtheorem{definition}[theorem]{Definition}

\theoremstyle{remark}
\newtheorem{remark}[theorem]{Remark}
\numberwithin{equation}{section}

\def\N{\mathbb{N}}
\def\R{\mathbb{ R}}
\def\Imm{\Im\mathrm{m}}
\def\Ree{\Re\mathrm{e}}
\def\tr{\mathrm {tr}}

\begin{document}

\title[KMS boundary conditions]{High temperature convergence of the KMS boundary conditions: The Bose-Hubbard model on a finite graph}
\author{Z. Ammari}
\address{Univ Rennes, [UR1], CNRS, IRMAR - UMR 6625, F-35000 Rennes, France.}
\email{zied.ammari@univ-rennes1.fr}
\author{A. Ratsimanetrimanana}
\address{BCAM - Basque Center for Applied Mathematics. Alameda de Mazarredo 14, E-48009 Bilbao, Spain.}
\email{aratsimanetrimanana@bcamath.org}

\date{April 15, 2019.}



\begin{abstract}
The Kubo-Martin-Schwinger condition  is a widely studied fundamental property in quantum statistical mechanics  which characterises  the thermal equilibrium states of quantum systems.  In the seventies, G.~Gallavotti and E.~Verboven, proposed an analogue to the KMS condition for classical mechanical systems and highlighted its  relationship with the Kirkwood-Salzburg equations and with the Gibbs equilibrium measures. In the present article, we prove that in a certain limiting regime of high temperature the classical  KMS condition can be derived from the quantum condition in the simple case of  the Bose-Hubbard dynamical system on a finite graph.  The main ingredients of the proof are  Golden-Thompson inequality,  Bogoliubov  inequality and semiclassical analysis.
\end{abstract}

\maketitle
\tableofcontents

\section{Introduction}
A $\mathscr{W}^*-$ dynamical system $(\mathscr{A},\tau_t)$  is a pair of a von Neumann  algebra of observables $\mathscr{A}$ and a one-parameter group of automorphisms $\tau_t$ on $\mathscr{A}$.  Consider for instance a finite dimensional Hilbert space $\mathfrak{H}$ then $\mathscr{A}$ can be chosen  to be the set of all operators  $\mathscr{B}(\mathfrak{H})$ and  $\tau_t$ to be the automorphism group defined by
$$
\tau_t(A)=e^{it H} A \,e^{-it H}
$$
for any $A\in\mathscr{A}$. The operator $H$ denotes the  Hamiltonian of a given quantum system and the couple $(\mathscr{A},\tau_t)$ describes  the dynamics.  According to quantum statistical physics such system admits a unique thermal equilibrium state  $\omega_\beta$ at inverse temperature $\beta$ given by,
\begin{equation}
\label{int.eq.3}
\omega_\beta (A) =\frac{\tr(e^{-\beta H} A)}{\tr(e^{-\beta H} )}\,.
\end{equation}
In general, the simplicity of the above statement  have to be nuanced.
In fact, the characterisation of thermal equilibrium in  statistical mechanics is a nontrivial question particularly for dynamical systems which have an infinite number of degrees of freedom, see \cite{MR1441540,MR0289084}. One of the important and most elegant characterisation of equilibrium states was noticed  by R.~Kubo, P.C.~Martin and J.~Schwinger in the late fifties. It is based in the following observations in finite dimension. In fact, one remarks by a simple computation that  the Gibbs state $\omega_\beta$ in \eqref{int.eq.3} satisfies for all $t\in \mathbb{R}$ and any $A,B\in\mathscr{A}$ the identity,
\begin{equation}
\label{int.eq.1}
\omega_\beta( A \;\tau_{t+i\beta}(B)) = \omega_\beta(  \tau_t(B)A)\,,
\end{equation}
where  $\tau_{t+i\beta}(\cdot)$ denotes an analytic extension of the automorphism
$\tau_t$ to complex times given by
$$
\tau_{t+i\beta}(B)=e^{(-\beta+it) H} B \,e^{(\beta-it) H} \,.
$$
More remarkable,  if one takes a state $\omega$ that satisfies the same condition as \eqref{int.eq.1}  then
$\omega$ should be the Gibbs state $\omega_\beta$ in \eqref{int.eq.3}. This indicates that the equation \eqref{int.eq.1}  singles out the thermal equilibrium states  among all possible states of a quantum system. In the  late sixties, R.~Haag, N.M.~Hugenholtz and  M.~Winnink suggested  the identity  \eqref{int.eq.1} as a criterion for equilibrium states and they named it the KMS boundary condition after Kubo, Martin and Schwinger \cite{MR0219283}. The subject of KMS states is bynow deeply studied  specially from an algebraic standpoint. For instance, various characterisation related to correlation inequalities and to variational principles have been derived (see e.g. \cite{MR628509,MR0252004,MR1441540}). Other perspectives have also been explored related for instance to the Tomita-Takasaki theory and to the Heck algebra and number theory (see e.g. \cite{MR1993788,MR0675642,MR1366621}).
 
\bigskip
 In the seventies, G.~Gallavotti and E.~Verboven, suggested an analogue to the KMS boundary condition \eqref{int.eq.1}  which is suitable for classical mechanical systems and highlighted its relationship with the Kirkwood-Salzburg equations and with the Gibbs equilibrium measures, see \cite{MR0449393}.  
 The derivation of such condition is based in the following heuristic argument. Consider a state $\omega_\hbar$ satisfying the KMS boundary condition 
 \begin{equation}
 \omega_\hbar\big(BA\big)= \omega_\hbar\big(A\,\tau_{i\hbar \beta}(B)\big)
 \end{equation} 
 at inverse temperature $\hbar\beta$,  where $\hbar$ refers to the reduced Planck constant. This relation yields
 \begin{equation}
 \omega_\hbar\bigg(\frac{AB-BA}{i\hbar}\bigg)= \omega_\hbar\bigg(A \; \frac{\tau_{i\hbar\beta}(B)-B}{i\hbar}\bigg)\,.
 \end{equation}
Assume for the moment  that the space $\mathfrak{H}=L^2(\mathbb{R}^d)$, so one can consider that  the Hamiltonian $H$ and the  observables $A,B$ are  given by $\hbar$-Weyl-quantized symbols (i.e., $H=h^{W,\hbar}$, 
$A=a^{W,\hbar}$ and $B=b^{W,\hbar}$ for some smooth functions $a$ and $b$ defined over the phase-space $\mathbb{R}^{2d}$). Then the semiclassical theory firstly tell us that 
\begin{equation}
\frac{AB-BA}{i\hbar} \underset{\hbar\to 0}{\longrightarrow} \{a,b\}\,,\qquad \text{and} \qquad  \frac{\tau_{i\hbar\beta}(B)-B}{i\hbar} 
\underset{\hbar\to 0}{\longrightarrow} \beta \,\{h,b\}\,,
\end{equation}
where $\{\cdot,\cdot\}$ is the Poisson bracket and  $h$ denotes the Hamiltonian of the corresponding classical system. Secondly, the quantum states $\omega_\hbar$ (or at least a subsequence) converge in a weak sense to a semiclassical probability measure  $\mu$ over $\mathbb{R}^{2d}$ when $\hbar\to 0$.  Therefore, the expected classical KMS condition that should in principle characterise the statistical equilibrium for classical mechanical systems is formally given by
\begin{equation}
\label{int.eq.2}
\mu\big( \{a,b\}\big)=\beta \,\mu\big( a \,\{h,b\}\big)\,,
\end{equation}
for any smooth functions $a,b$ on the phase-space $\mathbb{R}^{2d}$. Here the notation $\mu(f)=\int_{\mathbb{R}^{2d}} f(u) \,d\mu(u)$ is used. After the works \cite{MR0449393,MR0408670}, M.~Aizenman et al.~showed in  \cite{MR0443784}  that the condition  \eqref{int.eq.2} singles out thermal equilibrium states for infinite classical mechanical systems among all probability measures. In particular, the only measure $\mu$ satisfying \eqref{int.eq.2} in our example is the Gibbs measure defined with respect to the Lebesgue measure by the density,
\begin{equation}\label{rough_Gibbs}
\mu_{\beta}= \frac{1}{z(\beta)} \; e^{-\beta h(u)}\,,
\end{equation}
where $z(\beta)$ is a normalisation constant. Note that  the above Gibbs measure $\mu_{\beta}$ can also be  characterised as an equilibrium state  by means of variational methods and maximum entropy  properties or by correlation inequalities, see \cite{MR1441540}.  Nevertheless,  in this note we  focus only in  the KMS boundary conditions for classical and quantum systems.  In general, the derivation of the classical KMS boundary  condition \eqref{int.eq.2}  from the quantum one  is a non trivial and interesting question which  depends on the considered dynamical system.  In our opinion, the classical KMS condition is an elegant characterisation of statistical equilibrium which deserves more attention from PDE analysts. Although this condition  has been studied in some subsequent works  (see e.g.  \cite{MR0405121,MR761332,MR0469018,MR0449390,MR0441175,MR0484128}), it seems not largely known.

\bigskip
Our main purpose in this  note, is to provide a  rigorous and simple  proof for the  derivation of the classical KMS condition \eqref{int.eq.2} as a consequence of the relation \eqref{int.eq.1}  and the classical limit, $\hbar\to 0$, for the Bose-Hubbard dynamical system on a finite graph. The system we consider  is governed by a typical many-body quantum Hamiltonian which  can be written in terms of creations annihilations operators and which is restricted to a finite volume.  Our proof of convergence is based on the  Golden-Thompson inequality,  the Bogoliubov  inequality and the semiclassical analysis in the Fock space. Since the classical phase-space of the system considered  here  is finite dimensional it is possible by change of representation to convert  the problem to a  semiclassical analysis in a $L^2$ space. However, we avoid such a change as we  lose  most  of the interesting  insights and structures in our problem. In particular, we will rely on the analysis on the phase-space given in \cite{MR2465733}. Our interest in the Bose-Hubbard  system is motivated by the establishment of a strong link between classical and quantum KMS conditions so that it leads  to the  exchange of the  thermodynamic and the classical limits  for infinite dynamical systems and to the  investigation of phase transitions.   Also note that  from a physical standpoint, the Bose-Hubbard model  is  a quite relevant model describing  ultracold atoms in optical lattices with an observed phenomenon of superfluid-insulator transition.    From a wider perspective, the question considered here is also related to the recent trend initiated by M.~Lewin, P.T.~Nam and N.~Rougerie   \cite{MR3366672, MR3787331} about the  Gibbs measures for the nonlinear Schr\"odinger equations (see also  \cite{MR3719544} where these  investigations were continued). In this respect, the KMS boundary conditions could provide  an  alternative proof for the convergence of Gibbs states. These  questions will be considered elsewhere and here we will only focus on the Bose-Hubbard model on finite graph which is a much simpler model.

\bigskip
The article is organised as follows:
\begin{itemize}
\item In Section \ref{setup}, the Bose-Hubbard Hamiltonian on a finite graph is introduced and its relationship  with the discrete Laplacian is highlighted. 
\item Section \ref{KMS_micro}, is dedicated to the description of the unique KMS state of the Bose-Hubbard dynamical system at inverse  temperature $\hbar\beta$ and to the extension of the dynamics to complex times.   
\item Section \ref{convergence}, contains our main contribution stated in  Theorem \ref{KMSclassic}. Indeed, we prove that the KMS states  of the Bose-Hubbard system converge, up to subsequences,  to  semiclassical (Wigner) measures satisfying the classical KMS condition. The analysis is based on semiclassical methods in the Fock space developed in \cite{MR2465733}.
\item Finally, in Section \ref{sec:CKMS}, we remark that any probability measure satisfying the classical KMS condition 
is indeed  the Gibbs equilibrium measure for the Discret nonlinear Schr\"odinger equation. The proof of  this fact is   borrowed from the work \cite{MR0443784}. 
\end{itemize}

\section{Quantum Hamiltonian on a finite graph}\label{setup}
\emph{The discrete Laplacian:}  Consider a  finite   graph $G=(V,E)$ where $V$ is the set of vertices  and $E$ is the set of edges.  Assume furthermore that $G$ is a simple  undirected  graph and let $\deg(x)$ denotes the degree of each vertices $x\in V$. In the following, we denote  the graph equivalently $G$ or $V$.  Consider  the Hilbert space of all complex-valued functions on $V$ denoted as  $\ell ^{2} \left(  G \right)$ and 
endowed with its natural scalar product and with the orthonormal basis  $\left( e_{x} \right) _{x \in V} $ such that
$$
e_{x} (y) := \delta_{x,y}, \quad \forall x, y \in V.
$$ 
Then the  discrete Laplacian on the graph $G$ is a non-positive bounded operator on $\ell ^{2} \left(  G \right) $ given by,
$$
 \left(\Delta_{G} \psi \right) (x)  := - deg(x) \psi (x) + \sum_{y\in V,y \sim x} \psi (y),
$$
with the above sum running over the nearest neighbours of $x$ and $\psi$  is any function   in $\ell^{2} (G)$. 

\bigskip
\emph{The Bose-Hubbard Hamiltonian:}  Consider the bosonic Fock space,
$$
\mathfrak{F} = \mathbb{C} \oplus \bigoplus_{n=1}^{\infty} \otimes_{s}^n\ell ^{2}  \left(  G \right)^{}\,,
$$
where $\otimes_{s}^n\ell ^{2} \left(  G \right)^{}$ denotes the symmetric $n$-fold tensor product of $\ell ^{2}  \left(  G \right)$. So, any $\psi\in \otimes_{s}^n\ell ^{2}\left(  G \right)^{}$ is a functions $\psi:V^n\to \mathbb{C}$  invariant under any permutation of its variables. Introduce the usual  creation and annihilation operators acting on the bosonic Fock space,
$$
a_{x} = a(e_{x}) \quad \text{and} \quad a^{\ast}_{x} = a^{\ast}(e_{x})\,,
$$
then the following canonical commutation relations are satisfied,
$$
\left[ a_{x}, a^{*}_{y} \right] =  \,\delta_{x,y} \,\mathbf{1}_{\mathfrak{F}} \quad \text{and} \quad \left[ a^{*}_{x}, a^{*}_{y} \right] =\left[ a_{x}, a_{y} \right]= 0,  \quad \forall x,y \in V\,.
$$
\begin{definition}[Bose-Hubbard Hamiltonian]  \label{bosehubbard}
For $\varepsilon\in(0,\bar\varepsilon)$, $\lambda>0$ and  $\kappa<0$, define the $\varepsilon$-dependent Bose-Hubbard Hamiltonian on the bosonic Fock space $\mathfrak{F}$ by
$$
H_{\varepsilon}:=   \frac{\varepsilon}{2}  \sum_{x,y \in V: y \sim x }   (a^{*}_{x}- a^{*}_{y})  (a_{x} -a_{y})  +  \,\frac{\varepsilon ^{2} \lambda}{2} \sum_{x\in V} a^{*} _{x} a^{*} _{x} a _{x}a _{x} - \varepsilon \kappa \sum_{x\in V} a^{*}_{x} a_{x}.
$$
Here $\lambda$ is the on-site interaction, $\kappa$ is the chemical potential and $\varepsilon$ is the semiclassical parameter. 
\end{definition}

\begin{remark}
The first term of the Hamiltonian $H_{\varepsilon}$ is the kinetic part of the system and corresponds to the second quantization of the  discrete Laplacian. Indeed, one can write
 $$
\frac{1}{2}  \sum_{x,y \in V: y \sim x }   (a^{*}_{x}- a^{*}_{y})  (a_{x} -a_{y})=\sum_{x\in V} deg(x) \,a^*_x a_x - \sum_{x,y\in V, y\sim x} a^{*}_{x} a_{y}=
 \mathrm{d}\Gamma (- \Delta_{G})\,, 
$$
where $\mathrm{d}\Gamma (\cdot)$  is  the second quantization operator defined on the bosonic Fock space by 
\begin{equation}
\label{dgam}
\mathrm{d}\Gamma (A)_{|\otimes^n_s\ell^2(G)}= \sum_{j=1}^n 1\otimes \cdots\otimes A^{(j)} \otimes \cdots \otimes 1\,,
\end{equation}
for any given operator $A\in \mathscr{B}(\ell^2(G))$ and where $A^{(j)}$ means that $A$ acts only in the $j$-th component.  
\end{remark}

\medskip
\noindent 
The following rescaled \emph{number operator} will be often used,
\begin{equation}
\label{numberop}
N_\varepsilon:= \varepsilon\, \mathrm{d}\Gamma (1_{\ell^2(G)})= \varepsilon\, \sum_{x\in V} a^*_x a_x\,.
\end{equation}
Therefore, one can rewrite the Bose-Hubbard Hamiltonian  as follows
\begin{equation*}
H^{ \varepsilon} =  \varepsilon \,\mathrm{d}\Gamma \big(- \Delta_{G}-\kappa 1_{\ell^2(G)}\big) + \varepsilon^{2} \,\frac{\lambda}{2} \,I_{G} \,,
\end{equation*}
with   the interaction denoted as  
$$
I_{G}:= \sum_{x\in V} a^{*} _{x} a^{*} _{x} a _{x}a _{x}\,.
$$
Since the discrete Laplacian $\Delta_{G}$ is self-adjoint, it is easy to check that $H_{\varepsilon}$ defines an (unbounded) self-adjoint operator 
on the Fock space $\mathfrak{F}$ over its natural domain (for more details see e.g. \cite[Appendix A]{MR3379490}). Remark that the  operator $- \Delta_{G}-\kappa 1_{\ell^2(G)}$ is positive since the chemical potential $\kappa$ is negative.

\section{Quantum KMS condition} \label{KMS_micro}
The Bose-Hubbard Hamiltonian defines a $\mathscr{W}^*$ -dynamical system $(\mathfrak{M}, \alpha_t)$ where $\mathfrak{M}$ is the von Neumann algebra of all bounded operators $\mathscr{B}(\mathfrak{F})$ on the Fock space and $\alpha_t$ is the one parameter group of automorphisms defined by 
$$
 \alpha_t(A)= e^{i\frac{t}{\varepsilon} H_{\varepsilon}} \,A \,  e^{-i\frac{t}{\varepsilon}H_{\varepsilon}} \,,
$$
for any $A\in \mathfrak{M}$. The above group of automorphisms $\alpha_t$ admits a generator $S:\mathfrak{M}\to \mathfrak{M}$ with a domain  
$$
\mathcal{D}(S)=\{A\in\mathfrak{M}, \;[H_\varepsilon, A]\in\mathfrak{M}\}\,,
$$
and satisfies  for any $A\in \mathcal{D}(S)$, 
$$
S(A)=\lim_{t\to 0} \frac{\alpha_t(A) -A}{t}= \frac{i}{\varepsilon}[H_{\varepsilon}, A]\,.
$$
The latter convergence is with respect to the $\sigma$-weak topology on $\mathfrak{M}$. Remark also that the dynamics $\alpha_t$ depend on the semiclassical parameter $\varepsilon$. 

\bigskip
Next, we point out that the dynamical system $(\mathfrak{M}, \alpha_t)$ admits a unique KMS state  at inverse temperature $\varepsilon \beta$. Here $\beta>0$ is a fixed, $\varepsilon$-independent, effective inverse temperature. 
\begin{lemma}[Partition function]\label{class-trace}\mbox{}\newline \label{sec3.lem1}
Since the chemical potential $\kappa$ is negative  then 
$$
\mathrm{tr}_{\mathfrak{F}}\left( e^{-\beta H_{\varepsilon} }\right) < \infty.
$$
\end{lemma}
\begin{proof} 
It is a consequence of \cite[Proposition 5.2.27]{MR1441540} and the Golden-Thompson inequality. The latter, see \cite{MR3202886},  says that for any Hermitian matrices $A$ and $B$ one has,
\begin{equation}\label{thompson1}
\text{tr}\left( e^{A +B }\right) \leq \text{tr} \left( e^{A}\, e^{ B }\right)\,.
\end{equation}
\end{proof}
\begin{definition}[Gibbs state]\mbox{}\newline
The Gibbs equilibrium  state of the Bose-Hubbard system on a finite graph is well defined, according to Lemma \ref{sec3.eq.1},  and it is given by 
\begin{equation}
\label{sec3.eq.1}
\omega_\varepsilon (A) =\frac{\tr_{\mathfrak{F}}(e^{-\beta H_{\varepsilon}} A)}{\tr_{\mathfrak{F}}(e^{-\beta H_{\varepsilon}} )}\,.
\end{equation}
\end{definition}
For the sake of completeness, we recall some useful details concerning the KMS states.   One says that $A\in\mathfrak{M}$ is an \emph{entire  analytic element} of $ \alpha_t$ if 
there exists a function $f: \mathbb{C}\to \mathfrak{M}$ such that $f(t)= \alpha_t(A)$ for all $t\in\mathbb{R}$ and such that for any  trace-class operator $\rho\in \mathfrak{M}$ the function $z\in\mathbb{C}\to \tr(\rho f(z))$ is analytic. Let $\mathfrak{M}_{ \alpha}$ denotes the set of entire analytic elements for $ \alpha$, then it is known that 
$\mathfrak{M}_{ \alpha}$ is dense in $\mathfrak{M}$ with respect to the $\sigma$-weak topology. For more details on analytic elements, see \cite[section 2.5.3]{MR887100}. In particular, by \cite[Definition 2.5.20]{MR887100},  an element $A\in \mathfrak{M}$ is entire analytic  if and only if $A\in \mathcal{D}(S^n)$ for all $n\in\mathbb{N}$  and for any $t>0$ the series below are absolutely convergent,
\begin{equation}
\label{sec3.eq.2} 
\sum_{n=0}^\infty \frac{t^n}{n!} \, \| S^n(A)\|<\infty\,.
\end{equation}
Remark that on the set of entire analytic elements $\mathfrak{M}_\alpha$, the dynamics $\alpha_t$ can be extends to complex times. Indeed, 
$\alpha_z(A)$ is well defined, for any $A\in\mathfrak{M}_\alpha$, by the following absolutely convergent series,
$$
\alpha_z(A)= \sum_{n=0}^\infty \frac{z^n}{n!}  \;S^n(A)\,, \quad \forall z\in\mathbb{C}\,.
$$
We say that a state $\omega$ is a $(\alpha_t,\varepsilon\beta)$-KMS state if and only if $\omega$ is normal and for any $A,B\in \mathfrak{M}_{\alpha}$,
\begin{equation}
\label{sec3.KMS}
\omega( A \;\alpha_{i\varepsilon\beta}(B)) = \omega(  BA)\,.
\end{equation}
Remark that the above identity is known to be equivalent to the condition  stated in the introduction \eqref{int.eq.1}.  In particular, the KMS states are  stationary  states with respect to the dynamics. 
\begin{proposition} 
The Gibbs state $\omega_\varepsilon$ defined by \eqref{sec3.eq.1} is the unique KMS state of the $\mathscr{W}^*$ -dynamical system $(\mathfrak{M}, \alpha_t)$ at the inverse temperature $\varepsilon\beta$. 
\end{proposition}
\begin{proof}
 For $A,B\in \mathfrak{M}_{\alpha}$, one checks 
$$
\alpha_{i\varepsilon\beta}(B)= e^{-\beta H_\varepsilon} B e^{\beta H_\varepsilon}\,.
$$
The formula \eqref{sec3.eq.1} for the Gibbs state, gives 
$$
\omega_{\varepsilon}( A \;\alpha_{i\varepsilon\beta}(B)) =\frac{1}{\tr_{\mathfrak{F}}(e^{-\beta H_{\varepsilon}} )}  \tr_{\mathfrak{F}} \left(
A \,e^{-\beta H_\varepsilon} B\right) = \omega_\varepsilon( BA)\,.
$$ 
Reciprocally,  let $\omega$ be a $(\alpha_t,\varepsilon\beta)$-KMS state. In particular, there exists  a density matrix $\rho$ such that $\mathrm{tr}_{\mathfrak{F}}(\rho) =1$ and
$$
\omega(A) = \mathrm{tr}_{\mathfrak{F}}(\rho \, A), \quad \forall A \in \mathfrak{M}\,. 
$$
Using the KMS condition  \eqref{sec3.KMS} and the cyclicity of the trace, one proves for any $A\in\mathfrak{M}$,
$$
\mathrm{tr}(\rho \, B \,  A) = \mathrm{tr}( e^{-\beta H_{\varepsilon}} \, B \, e^{\beta H_{\varepsilon}}\rho \, A ).
$$
In particular, for any $B\in\mathfrak{M}_\alpha$,
\begin{equation}\label{egalite3}
 \rho \,B ={e}^{-\beta H_{\varepsilon}} B \, {e}^{\beta H_{\varepsilon}}\rho \, .
\end{equation}
Hence, one  remarks that $\rho$ commutes with any spectral projection of $H_\varepsilon$ by taking for instance $B=1_{D}(H_\varepsilon)$  in the equation \eqref{egalite3}.  Therefore, one concludes that  
\begin{equation*}
\label{egalite4}
{e}^{\beta H_{\varepsilon}} \, \rho \,B\,\underset{| 1_{D}(H_\varepsilon)\mathfrak{F}}{}=\,B \, e^{\beta H_{\varepsilon}}\rho \,\underset{| 1_{D}(H_\varepsilon)\mathfrak{F}}{}\;,
\end{equation*}
for any bounded Borel subset $D$ of $\mathbb{R}$ and any bounded operator $B$ satisfying $B =1_{D}(H_\varepsilon) B= B 1_{D}(H_\varepsilon)$.  So, the operator ${e}^{\beta H_{\varepsilon}}\, \rho$ commutes with any bounded operator over the subspaces $1_{D}(H_\varepsilon)\mathfrak{F}$. This implies that
$$
\rho = c \, e^{-\beta H_{\varepsilon}}\,,
$$
and then one concludes with the fact that $\mathrm{tr}(\rho) =1$.
\end{proof}

\section{Convergence}
\label{convergence}
In this section, we prove that the KMS condition \eqref{sec3.KMS}  converges, in the classical limit,   towards the classical KMS condition.   It is enough to prove such convergence for some specific observables $A,B\in\mathfrak{M}$. In fact, consider for $f,g\in \ell^2(G)$, 
\begin{equation}
\label{sec4.eq.AB}  
A=W(f)\,, \qquad \text{ and } \qquad  B=W(g)\,,
\end{equation}
where $W(\cdot)$ denotes the Weyl operator defined by,
\begin{equation}
\label{weyl}
W(f)=e^{i \sqrt{\varepsilon}\;\Phi(f)}\,, \qquad \text{ with  }  \quad \Phi(f)=\frac{a^*(f)+a(f)}{\sqrt{2}}\,.
\end{equation}
Let $\chi\in\mathscr{C}_0^\infty(\R)$ such that $0\leq \chi\leq 1$,  $\chi\equiv 1$ if $|x|\leq 1/2$ and $\chi\equiv 0$ if $|x|\geq 1$. Define, for $n\in\N$,  the cut-off functions $\chi_n$ as 
$$
\chi_n(\cdot)=\chi\big(\frac{\cdot}{n}\big)\,.
$$
Then, we are going to  consider only the following  smoothed observables,
\begin{eqnarray}
\label{sec4.eq.1}
A_n:= \chi_n(N_\varepsilon) \, A\,  \chi_n(N_\varepsilon)\,, \qquad 
\text{ and } \qquad B_n:= \chi_n(N_\varepsilon) \, B\,  \chi_n(N_\varepsilon)\,.
\end{eqnarray}
\begin{lemma}
\label{lem.smoothing}
For all $ \varepsilon>0$ and $n\in\N$, the elements $A_n$ and  $B_n$ given by \eqref{sec4.eq.1}  are entire analytic for the dynamics $\alpha_t$. 
\end{lemma}
\begin{proof}
By functional calculus, remark that $1_{[0,n]}(N_\varepsilon) \chi_n(N_\varepsilon)= \chi_n(N_\varepsilon)$. Moreover, 
the number operator $N_\varepsilon$ and the Hamiltonian  $H_\varepsilon$ commute in the strong sense. So,  
the generator $S$  of the dynamics $\alpha_t$ satisfies for $k\in\N$,
\begin{eqnarray*}
S^{k}(A_n)& =& \left(\frac{i}{\varepsilon} \right)^{k} \, 
[H_\varepsilon,\cdots[H_\varepsilon, A_n]\cdots]\,, \\
& =& \left(\frac{i}{\varepsilon} \right)^{k} \, 
[\tilde H_\varepsilon,\cdots[\tilde H_\varepsilon, A_n]\cdots]\,,
\end{eqnarray*}
with $\tilde H_\varepsilon= 1_{[0,n]}(N_\varepsilon) H_\varepsilon$ a  bounded operator. Hence, the estimate \eqref{sec3.eq.2} is satisfied and so $A_n$ is a entire analytic element. 
\end{proof}
Recall that the $(\alpha_t,\varepsilon\beta)$-KMS state $\omega_\varepsilon$ satisfies in particular the condition,
$$
\omega_\varepsilon\left( A_n \, \alpha_{i\varepsilon\beta}(B_m)\right)=
\omega_\varepsilon\left( B_m \,A_n \right)\,.
$$
A simple computation then leads to the main identity,
\begin{equation}
\label{sec4.eq.main}
\omega_\varepsilon\left( A_n \, \frac{\alpha_{i\varepsilon\beta}(B_m)-B_m}{i\varepsilon}\right)=
\omega_\varepsilon\left( \frac{[B_m,A_n]}{i\varepsilon}\right)\,.
\end{equation}
Our aim  is to take  the  classical limit $\varepsilon\to 0$   in the above relation and to prove the  convergence for the left and right hand sides so that we obtain the classical KMS boundary conditions.  In order to take such limit, one needs to use the  semiclasscial (Wigner) measures of $\{\omega_\varepsilon\}_{\varepsilon\in(0,\bar\varepsilon)}$.  Recall that $\mu$ a Borel probability measure on the phase-space $\ell^2(G)$ is a Wigner measure of  $\{\omega_\varepsilon\}_{\varepsilon\in(0,\bar\varepsilon)}$ if there exists a subsequence $(\varepsilon_k)_{k\in\N}$ such that $\lim_{k\to \infty}\varepsilon_k= 0$ and for any $f\in\ell^2(G)$, 
\begin{equation}
\label{sec4.eq.wigner}
\lim_{k\to \infty}  \omega_{\varepsilon_k}\left( W(f)\right)= \int_{\ell^2(G)}  e^{i\sqrt{2} \Ree\langle f,u\rangle} \;d\mu\,.
\end{equation}
Note that the Weyl operator depends here on the parameter $\varepsilon_k$ instead of $\varepsilon$ as in  \eqref{weyl}.   According  to \cite[Thm.~6.2]{MR2465733}   and Lemma \ref{particle_number}, the family of KMS states $\{\omega_\varepsilon\}_{\varepsilon\in(0,\bar\varepsilon)}$ admits a non-void set of Wigner probability measures. Later on, we will prove that this set of measures reduces to a singleton given by the Gibbs equilibrium  measure. But for the moment, we will use subsequences as in the definition \eqref{sec4.eq.wigner}.

\medskip

 The classical Hamiltonian system related to the Bose-Hubbard model  is given by the \emph{Discrete Nonlinear Schr\"odinger}  equation, see \cite{MR2742565}. Its energy functional  (or Hamiltonian )  is given by 
 \begin{equation}
 \label{hamdnls}
 h(u)=-\langle u, \Delta_G \,u\rangle -\kappa \|u\|^2+\frac{\lambda}{2}\sum_{j\in V} |u(j)|^4\,.
 \end{equation}
Note that $\ell^2(G)$ is  a complex Hilbert space and  so in our framework the Poisson structure is defined as follows. 
For $F,G$ smooth functions on $\ell^{2}(G)$, the Poisson bracket is given by
 \begin{equation}\label{poisson}
 \left\lbrace F,G \right\rbrace := \frac{1}{i} \; \left(\partial_{u} F\cdot \partial_{\bar u} G- \partial_{u} G\cdot \partial_{\bar u} F\right)\,.
 \end{equation}
Here $\partial_{u}$ and $\partial_{\bar u}$ are the standard differentiation with respect to $u$ or $\bar u$.  

\medskip
Our main result is stated below. 
\begin{theorem}[Classical KMS condition]
 \label{KMSclassic}
Let $\omega_\varepsilon$ by the KMS state of the Bose-Hubbard $\mathscr{W}^*$-dynamical system $(\mathscr{A},\alpha_t)$ at inverse temperature $\varepsilon \, \beta$. Then any semiclassical (Wigner) measure of $\omega_\varepsilon$ satisfies the classical KMS condition, i.e., for any $F,G$ smooth functions on $\ell^{2}(G)$, 
\begin{equation}
\label{CKMS}
 \beta\; \mu(\left\lbrace h, G\right\rbrace \, F )=\mu(\left\lbrace F,G\right\rbrace )\,,
\end{equation}
where the classical Hamiltonian $h$ is given by \eqref{hamdnls} and $\{\cdot , \cdot\}$ denotes the Poisson bracket recalled in \eqref{poisson}. 
\end{theorem}

In order to prove Theorem \ref{KMSclassic}, one needs some preliminary steps. 

\begin{proposition}
\label{sec4.prop.2}
Let $(\varepsilon_k)_{k\in\N}$ be a subsequence such that $\lim_{k\to \infty}\varepsilon_k= 0$. Assume that the  family of KMS states $\{\omega_{\varepsilon_k}\}_{k\in\N}$ admits a unique Wigner measure $\mu$. Then for all $n,m$  integers such that $m\geq 2 n$,
\begin{eqnarray}
\label{sec4.eq.3}
\lim_{k\to\infty} \omega_{\varepsilon_k}\left( \frac{[B_m,A_n]}{i\varepsilon_k} \right)&=&\int_{\ell^2(G)}  
\chi_n^2(\langle u, u\rangle) \; \{ e^{\sqrt{2}i  \Ree\langle g,u\rangle}; e^{\sqrt{2}i  \Ree\langle f,u\rangle}\}
 \;d\mu\,\\ \label{sec4.eq.4}
 &&+  
\int_{\ell^2(G)}  
\chi_n(\langle u, u\rangle) \; \{ e^{\sqrt{2}i  \Ree\langle g,u\rangle}; \chi_n(\langle u, u\rangle) \} \; e^{\sqrt{2}i  \Ree\langle f,u\rangle}
 \;d\mu\\
 \label{sec4.eq.5}
  && + \int_{\ell^2(G)}  
\chi_n(\langle u, u\rangle) \; \{ \chi_n(\langle u, u\rangle); e^{\sqrt{2}i  \Ree\langle f,u\rangle}\}\;e^{\sqrt{2}i  \Ree\langle g,u\rangle}
 \;d\mu\,.
\end{eqnarray} 
\end{proposition}
\begin{proof}
For simplicity, we denote $\varepsilon$ instead of  $\varepsilon_k$ and  $\chi_m$ instead of $\chi_m(N_{\varepsilon})$. Using the cyclicity of the trace and the fact that $\chi_n\chi_m=\chi_n$, one remarks   that
$$
\omega_{\varepsilon}\left( [B_m,A_n] \right)= \omega_{\varepsilon} \left(\chi_n \, (B \chi_n A- A\chi_n B) \right)\,.
$$
A simple computation yields,
\begin{equation}
\label{sec4.eq.6}
\begin{aligned}
\lim_{\varepsilon\to 0} \omega_{\varepsilon}\left(  \frac{[B_m,A_n]}{i\varepsilon} \right)&=  \lim_{\varepsilon\to 0}\omega_{\varepsilon} \left(\chi_n^2 \,  \frac{[B,A]}{i\varepsilon}  \right)\\
&+  \lim_{\varepsilon\to 0}\omega_{\varepsilon} \left(\chi_n \,  \frac{[B,\chi_n]}{i\varepsilon}  \,A\right)\\
&+  \lim_{\varepsilon\to 0}\omega_{\varepsilon} \left(\chi_n \,  \frac{[\chi_n,A]}{i\varepsilon}  \, B\right)\,.
\end{aligned}
\end{equation}
The Weyl commutation relations give,
$$
\frac{[B,A]}{i\varepsilon}=W(f+g) \; (\Imm\langle f,g\rangle+O(\varepsilon))\,.
$$
So, using Lemma \ref{apx.lem-est2}, 
\begin{eqnarray*}
 \lim_{\varepsilon\to 0}\omega_{\varepsilon} \left(\chi_n^2 \,  \frac{[B,A]}{i\varepsilon}  \right)&=&
\Imm\langle f,g\rangle  \;\lim_{\varepsilon\to 0}\omega_{\varepsilon} \left(\chi_n^2 \,  W(f+g) \right)\\
&=&
\Imm\langle f,g\rangle  \int_{\ell^2(G)}  
\chi_n^2(\langle u, u\rangle) \; e^{\sqrt{2}i  \Ree\langle f+g,u\rangle}  \;d\mu\,.
\end{eqnarray*}
Checking the Poisson bracket,
$$
 \{ e^{\sqrt{2}i  \Ree\langle g,u\rangle}; e^{\sqrt{2}i  \Ree\langle f,u\rangle}\}= \Imm\langle f, g\rangle \; e^{\sqrt{2}i  \Ree\langle f+g,u\rangle}\,,
$$
one obtains the right hand side of \eqref{sec4.eq.3}.   Consider now the second term in \eqref{sec4.eq.6}. One can write 
$$
[W(g), \chi_n]=\int_{\R} \hat\chi_n(s) \, [W(g), e^{is N_\varepsilon}] \, \frac{ds}{\sqrt{2\pi}}\,,
$$
where $\hat\chi_n$ denotes the Fourier transform of the function $\chi_n(\cdot)$. Using standard computations in the 
Fock space and Taylor expansion, 
\begin{eqnarray*}
[W(g), e^{is N_\varepsilon}]&=& e^{is N_\varepsilon} \left( e^{-is N_\varepsilon} W(g) e^{is N_\varepsilon} -W(g)\right)\\
&=& i e^{is N_\varepsilon}  \;\int_0^s e^{-ir N_\varepsilon} [W(g), N_\varepsilon] \, e^{ir N_\varepsilon} \, dr\\
&=& - e^{is N_\varepsilon}  \;\int_0^s e^{-ir N_\varepsilon} W(g) \left(\varepsilon \Phi(ig)+ \frac{\varepsilon^2}{2} \|g\|^2\right) \, e^{ir N_\varepsilon} \, dr \,.
\end{eqnarray*}
Hence, using the cyclicity of the trace
\begin{equation}
 \lim_{\varepsilon\to 0}\omega_{\varepsilon} \left(\chi_n \,  \frac{[B,\chi_n]}{i\varepsilon}  \,A\right)= 
 -\int_{\R}  s \hat\chi_n(s) \, \lim_{\varepsilon\to 0}\omega_{\varepsilon} \left(\chi_n \, e^{is N_\varepsilon} W(g) \Phi(ig) W(f)\right) \; \frac{ds}{\sqrt{2\pi}}\,.
\end{equation}
Knowing, by Lemma \ref{apx.lem-est2}, that the Wigner measure of  the sequence  $\{W(f) \rho_\varepsilon \chi_n(N_\varepsilon) e^{is N_\varepsilon} W(g)\}$ is given by 
$$
\big\{\mu \chi_n(\langle u, u\rangle) e^{i s \|u\|^2} e^{\sqrt{2} i \Ree\langle g+f,u\rangle}  \big\}\,,
$$
then one obtains  using \cite[Thm.~6.13]{MR2465733}, 
$$
\lim_{\varepsilon\to 0}\omega_{\varepsilon} \left(\chi_n \,  \frac{[B,\chi_n]}{i\varepsilon}  \,A\right)=  
-\sqrt{2}  \int_\R s \hat\chi_n(s) \, \int_{\ell^2(G)}  \chi_n(\langle u,u\rangle) 
 \, e^{is \|u\|^2}  \Ree\langle u, i g\rangle e^{\sqrt{2} i \Ree\langle g+f,u\rangle} d\mu \; \frac{ds}{\sqrt{2\pi}}\,.
$$
Integrating back with respect to the variable $s$, 
\begin{eqnarray*}
\lim_{\varepsilon\to 0}\omega_{\varepsilon} \left(\chi_n \,  \frac{[B,\chi_n]}{i\varepsilon}  \,A\right)= 
\sqrt{2}i \int_{\ell^2(G)} \chi_n'(\|u\|^2) \,  \chi_n(\|u\|^2)   \,  \Imm\langle g,u\rangle 
e^{\sqrt{2} i \Ree\langle g+f,u\rangle}  \;d\mu\,.
\end{eqnarray*}
Then checking the Poisson bracket 
$$
 \big\{ e^{\sqrt{2}i  \Ree\langle g,u\rangle}; \chi_n(\langle u, u\rangle) \big\}=
 \sqrt{2}i \chi_n'(\|u\|^2) \;\Imm\langle g, u\rangle e^{\sqrt{2}i  \Ree\langle g,u\rangle}\,,
$$
yields the right hand side of \eqref{sec4.eq.4}. The third term in the right side of \eqref{sec4.eq.6} is similar to the above one. 
\end{proof}

The next step is to prove the convergence of the left hand side of \eqref{sec4.eq.main}. 

\begin{lemma}
\label{lem.tech.1}
\begin{eqnarray}
\label{sec.eq.4}
\lim_{k\to \infty} \omega_{\varepsilon_k}\left(A_n \;  \frac{\alpha_{i\varepsilon_k\beta}(B_m)-B_m}{i\varepsilon_k}\right)=
\beta\, \lim_{k\to \infty} \omega_{\varepsilon_k}\left( A_n \;  \frac{ [B_m, H_{\varepsilon_k}]}{i\varepsilon_k}\right)\,.
\end{eqnarray}
\end{lemma}
\begin{proof}
For simplicity, we use $\varepsilon$ instead of $\varepsilon_k$. According to Lemma \ref{lem.smoothing}, $B_m$ is a entire analytic element for the dynamics $\alpha_t$. Hence, by Taylor expansion,
\begin{eqnarray*}
 \omega_{\varepsilon}\left(A_n \;  \frac{\alpha_{i\varepsilon\beta}(B_m)-B_m}{i\varepsilon}\right)=\beta\,
\int^1_0 \omega_{\varepsilon} \left(A_n \frac{[\alpha_{is\varepsilon\beta} (B_m), H_\varepsilon]}{i\varepsilon}\right) \; ds\,. 
\end{eqnarray*}
Using the cyclicity of the trace and the fact that $A_n$, $B_m$ are entire analytic elements, 
\begin{equation*}
\omega_{\varepsilon} \left(A_n \frac{[\alpha_{is\varepsilon\beta} (B_m), H_\varepsilon]}{i\varepsilon}\right) =\omega_{\varepsilon} \left(e^{s \beta H_\varepsilon}  \,A_n\,e^{-s \beta H_\varepsilon}   \;\frac{[B_m, H_\varepsilon]}{i\varepsilon}\right) \,.
\end{equation*}
A second Taylor expansion yields,
\begin{eqnarray*}
\omega_{\varepsilon} \left(A_n \,\frac{[\alpha_{is\varepsilon\beta} (B_m), H_\varepsilon]}{i\varepsilon}\right) &=&\omega_{\varepsilon} \left(  A_n  \;\frac{[B_m, H_\varepsilon]}{i\varepsilon}\right) \\ &&+\beta
\int_0^s \omega_\varepsilon\left( e^{r \beta H_\varepsilon}  \,[H_\varepsilon,A_n]\,e^{-r \beta H_\varepsilon}   \;\frac{[B_m, H_\varepsilon]}{i\varepsilon}\right) \, dr\,.
\end{eqnarray*}
So, the equality \eqref{sec.eq.4} is proved since 
$$
\lim_{\varepsilon\to 0} \int_0^1 ds \int_0^s dr \;\; \omega_\varepsilon\left(  \,[H_\varepsilon, \alpha_{-is\varepsilon\beta }(A_n)]  \;\frac{[B_m, H_\varepsilon]}{i\varepsilon}\right)\,=0\,,
$$
thanks to the Lemma \ref{apx.lem-est1} in the Appendix.  
\end{proof}

\begin{proposition}
\label{sec4.prop.1}
Let $(\varepsilon_k)_{k\in\N}$ be a subsequence such that $\lim_{k\to \infty}\varepsilon_k= 0$. Assume that the  family of KMS states $\{\omega_{\varepsilon_k}\}_{k\in\N}$ admits a unique Wigner measure $\mu$. Then for all $n,m$  integers such that $ m\geq 2 n$,
\begin{eqnarray}
\label{sec4.eq.2}
\lim_{k\to\infty} \omega_{\varepsilon_k}\left( A_n \, \frac{\alpha_{i\varepsilon\beta}(B_m)-B_m}{i\varepsilon}\right)=\beta\int_{\ell^2(G)}  
\chi_n^2(\langle u, u\rangle) \; \{ e^{\sqrt{2}i \Ree\langle g,u\rangle}; h(u)\} \,e^{\sqrt{2}i \Ree\langle f,u\rangle}
 \;d\mu\,.
\end{eqnarray} 
\end{proposition}
\begin{proof}
The previous Lemma \ref{lem.tech.1} allowed to get rid of the dynamics at complex times. So, it is enough to show the limit, 
$$
\lim_{k\to\infty}\omega_{\varepsilon_k}\left( A_n \;  \frac{ [B_m, H_{\varepsilon_k}]}{i\varepsilon_k}\right)=
\int_{\ell^2(G)}  
\chi_n^2(\langle u, u\rangle) \; \{ e^{\sqrt{2}i \Ree\langle g,u\rangle}; h(u)\} \,e^{\sqrt{2}i \Ree\langle f,u\rangle}
 \;d\mu\,.
$$
For simplicity, we denote $\varepsilon$ instead of $\varepsilon_k$ and  $\chi_m$ instead of $\chi_m(N_{\varepsilon})$. Since $m\geq 2n$ then $\chi_n\chi_m=\chi_n$ and one notices that  
$$
\chi_n  A \chi_n \; [\chi_m B \chi_m , H_\varepsilon] = \chi_n A \chi_n \;[B,H_\varepsilon]\,\chi_m= \chi_n A \chi_n \;[W(g),H_\varepsilon]\,\chi_m\,. 
$$
Standard computations on the Fock space yield, (see e.g.~\cite[Proposition 2.10]{MR2465733}),
 \begin{eqnarray*}
 \frac{i}{\varepsilon}  [B,H_\varepsilon]&=& \frac{i}{\varepsilon}  \left( W(g) H_\varepsilon W(g)^*-H \right) \, W(g)\\
 &=& \frac{i}{\varepsilon} \left( h(\cdot-\frac{i\varepsilon}{\sqrt{2}} g)-h(u)\right)^{Wick} \;W(g) \\
 &=& \left(\underset{C^{Wick}}{\underbrace{\{\sqrt{2}\Ree\langle g, u\rangle, h(u)\}}}^{Wick}+ R(\varepsilon)^{Wick}\right)\; W(g)\,.
 \end{eqnarray*}
The subscript \emph{Wick} refers to the Wick quantization, see \cite[section 2]{MR2465733}.  The remainder $R(\varepsilon)^{Wick}$ can be explicitly computed and satisfies the uniform estimate  
$$
\|\chi_n(N_\varepsilon) \; R(\varepsilon)^{Wick}\|\leq c \;\varepsilon\,,
$$
which can be easily proved using \cite[Lemma 2.5]{MR2465733}.  Therefore, by using Lemma \ref{apx.lem-est1} one shows
\begin{eqnarray*}
\lim_{k\to\infty} \omega_{\varepsilon_k}\left( A_n \, \frac{\alpha_{i\varepsilon\beta}(B_m)-B_m}{i\varepsilon}\right)&=&\beta \lim_{k\to\infty} \omega_{\varepsilon_k}\left( \chi_n A \chi_n  \,  C^{Wick} B \right)\\
&=&\beta \lim_{k\to\infty} \omega_{\varepsilon_k}\left( \chi_n^2 A \,  C^{Wick} B\right)\,.
\end{eqnarray*}
Knowing, by Lemma \ref{apx.lem-est2}, that the Wigner measure  of the sequence  $\{ W(g) \rho_\varepsilon \chi^2_n(N_\varepsilon) W(f)\}$ is given by 
$$
\big\{\mu e^{\sqrt{2}i \Ree\langle f+g,u\rangle} \chi_n^2(\|u\|^2)\big\}\,,
$$
one concludes by  \cite[Thm.~6.13]{MR2465733}, 
$$
\lim_{\varepsilon\to 0} \omega_{\varepsilon}\left( A_n \, \frac{\alpha_{i\varepsilon\beta}(B_m)-B_m}{i\varepsilon}\right)= \int_{\ell^2(G)}   
\chi_n^2(\|u\|^2) \; e^{\sqrt{2}i \Ree\langle f+g,u\rangle} C(u) \;d\mu\,.
$$ 
\end{proof}

\begin{cor}
\label{sec4.cor.1}
Any Wigner measure of the $(\alpha_t,\varepsilon\beta)$-KMS family of states $\omega_{\varepsilon}$ satisfies for all $f,g\in\ell^2(G)$,
\begin{eqnarray}
\label{sec4.eq.7}
\beta\; \int_{\ell^2(G)}  
 \{ e^{i \Ree\langle g,u\rangle}; h(u)\} \,e^{i \Ree\langle f,u\rangle}
 \;d\mu\,
=\int_{\ell^2(G)}   \; \{ e^{i \Ree\langle g,u\rangle}; e^{i \Ree\langle f,u\rangle}\}
 \;d\mu\,.
 \end{eqnarray} 
\end{cor}
\begin{proof}
It is a consequence of Proposition \ref{sec4.prop.2}, Proposition \ref{sec4.prop.1}  and dominated convergence while taking $n,m\to \infty$.  
\end{proof}

\medskip
Thus, we come to the following conclusion.
\begin{proof}[Proof of Theorem \ref{KMSclassic}] \;\;
The phase-space $\ell^2(G)$ is a $d$-euclidean space. Let $F,G$ be two smooth functions in $\mathscr{C}_0^\infty(\ell^2(G))$. The inverse Fourier transform gives,
$$
F(u)=\int_{\ell^2(G)} e^{i \Ree\langle f,u\rangle} \; \hat F(f) \; \frac{dL(f)}{(2\pi)^{d/2}}\,, \qquad \text{ and } \qquad 
G(u)=\int_{\ell^2(G)} e^{i \Ree\langle g,u\rangle} \; \hat G(g) \; \frac{dL(g)}{(2\pi)^{d/2}}\,,
$$
where $\hat F, \hat G$ denote the Fourier transforms of $F$ and $G$ respectively.   
Multiplying the equation \eqref{sec4.eq.7}  by $ \hat F(f)  \hat G(g)$ and integrating with respect to   the Lebesgue measure in the variables $f$ and $g$,  one obtains
$$
\beta\;   \int_{\ell^2(G)}  
\{ G(u), h(u)\} \,F(u) \;d\mu\,
=\int_{\ell^2(G)}   \; \{ G(u), F(u)\}  \;d\mu\,.
$$
This proves the classical KMS condition \eqref{CKMS}. 
\end{proof}

\section{Classical KMS condition}
\label{sec:CKMS}
In this section, we point out  that the only probability measure  satisfying the classical KMS condition  is the Gibbs equilibrium measure. This is a known fact and we  provide here a short proof only for reader's convenience. The argument used below is borrowed from the work of M.~Aizenman, S.~Goldstein, C.~Grubber, J.~Lebowitz   and 
P.A.~Martin \cite{MR0443784}.

\begin{proposition}[Gibbs measure]
 \label{ClassicalGibbs}
Suppose that $\mu$ is a Borel probability measure on $\ell^2(G)$ satisfying the classical KMS condition \eqref{CKMS}. Then $\mu$ is the Gibbs equilibrium measure, i.e.,
$$
\frac{d\mu}{dL}= \frac{e^{-\beta \, h(u)} }{z(\beta) }\,, \quad \text{ and } \quad z(\beta) =\int_{\ell^2(G)} {e}^{-\beta \, h(u)} \,dL(u)\,,
$$
with $h(\cdot)$ is the classical Hamiltonian of the Discrete Nonlinear Schr\"odinger equation given by \eqref{hamdnls} and $dL$ is the Lebesgue measure on $\ell^2(G)$.
\end{proposition}
\begin{proof}
Consider the Borel probability measure  $\nu= e^{\beta h(u)} \mu$, so that  for any Borel set $\mathcal{B}$, 
$$
\nu(\mathcal{B})=\int_{\mathcal B} e^{\beta h(u)} d\mu\,.
$$
 Note that, for any $F,G\in \mathscr{C}_{0}^\infty(\ell^2(G))$, the Poisson  bracket satisfies
$$
\big\{F e^{-\beta h(u)}, G \big\}= \big\{F, G \big\}\, e^{-\beta h(u)}-\beta \big\{h, G\big\} \, F(u) e^{-\beta h(u)}\,.
$$
Hence, the classical KMS condition \eqref{CKMS} can be written as 
$$
\mu\left( e^{\beta h(u)} \big\{F e^{-\beta h(u)}, G \big\}\right)=0\,,
$$
or equivalently for any $F,G\in \mathscr{C}_{0}^\infty(\ell^2(G))$,
$$
\nu\left( \big\{F e^{-\beta h(u)}, G \big\}\right)=0\,.
$$
Remark that the classical Hamiltonian $h$ is a smooth $\mathscr{C}^\infty(\ell^2(G))$ function. Hence, the measure $\nu$ satisfies for any $F,G\in \mathscr{C}_0^\infty(\ell^2(G))$, 
$$
\nu\left( \big\{F, G \big\}\right)=0\,.
$$
This condition implies that $\nu$ is a multiple of the Lebesgue measure. Indeed, take 
$g(\cdot)=\langle e_j,  \cdot \rangle\; \varphi(\cdot)$ with $\varphi\in  \mathscr{C}_0^\infty(\ell^2(G))$ being equal to $1$ on the support of $f$. Then the Poisson bracket gives,
$$
\{ f,g\}=- i \partial_j f\,.
$$
So, in a distributional sense the derivatives of the measure $\nu$ are null and therefore $d\nu=c \,dL$ for some constant $c$.  Using the normalisation requirement for $\mu$, one concludes that  $d\nu= \frac{1}{z(\beta)} \,dL$. 
\end{proof}

\bigskip

\appendix
\section{Number estimates}
Consider the quasi free state ${\omega}^0_{\varepsilon}(\cdot)$  given by,
$$
{\omega}^0_{\varepsilon}(\cdot) = \frac{\tr \left(\cdot \;\,{e}^{\beta  \varepsilon\mathrm{d}\Gamma(\Delta_{G}+\kappa 1)}\right) }{\tr\left( e^{\beta \varepsilon \mathrm{d}\Gamma(\Delta_{G}+\kappa 1)}\right) }\;.
$$
The following uniform number of particles estimates are well know. Here we recall them for reader's convenience.  For more details on quasi free states and such inequalities, see  e.g. \cite{MR1441540,MR3366672,MR3719544}.  Remember that the rescaled number operator is given by,
$$
N_{\varepsilon}:=\varepsilon {\textrm d}\Gamma\left(1_{\ell^2(G)}\right)=\varepsilon\sum_{x \in V} a_{x}^{*}a_{x}.
$$
\begin{lemma} \label{particle_number0}
For any $k\in\N$, there exists a positive constant $c_k$ such that
$$
\omega^0_{\varepsilon} (N^k_{\varepsilon}) \leq c_k\,,
$$
uniformly with respect to $\varepsilon\in(0,\bar\varepsilon)$.
\end{lemma}

\begin{lemma}  \label{particle_number1}
There exists a positive constant $c$ such that
$$
\frac{\tr({e}^{\beta \varepsilon \mathrm{d}\Gamma(\Delta_{G}+\kappa 1)  })}{\tr( e^{-\beta H_{\varepsilon  }  })
} \leq c\,,
$$
uniformly with respect to  $\varepsilon\in(0,\bar\varepsilon)$.
\end{lemma}
\begin{proof}
By using a Bogoliubov type inequality, see \cite[Appendix D]{MR1853206}, one has that
$$
\ln(\tr(e^{\beta \varepsilon \mathrm{d}\Gamma(\Delta_{G}+\kappa 1)  })) - \mathrm{ln} (\tr( e^{-\beta H_{\varepsilon  }  })) \leq \beta \;\frac{ \tr\left(  \varepsilon^{2} \, \frac{\lambda}{2} I_{G}   \;e^{\beta \varepsilon \mathrm{d}\Gamma(\Delta_{G}+\kappa 1)  }\right) }{\mathrm{tr}\left( e^{\beta \varepsilon \mathrm{d}\Gamma(\Delta_{G}+\kappa 1)  }\right) }.
$$
According to Definition \ref{bosehubbard}, recall that 
$$
I_{G} = \sum_{x \in V} a^{*} _{x} a^{*} _{x} a _{x}a _{x} \,.
$$
Therefore, there exists $c>0$ such that 
$$
\ln (\tr(e^{\beta \varepsilon \mathrm{d}\Gamma(\Delta_{G}+\kappa 1)  })) -\ln (\mathrm{tr}( e^{-\beta H^{ \varepsilon  }  })) \leq c  \left( {\omega }^0_{\varepsilon}(N_{\varepsilon}^2)+{\omega }^0_{\varepsilon} (N_{\varepsilon})\right) .
$$
Using  Lemma \ref{particle_number0}, one proves the inequality. 
\end{proof}

\begin{lemma}\label{particle_number}
For any $k\in\N$, there exists a positive constant $c_k$ such that
$$
\omega_{\varepsilon} (N^k_{\varepsilon}) \leq c_k\,,
$$
uniformly with respect to $\varepsilon\in(0,\bar\varepsilon)$.
\end{lemma}
\begin{proof}
A direct consequence of Lemma \ref{particle_number0}, Lemma \ref{particle_number1} and the Golden-Thompson inequality. 
\end{proof}

\section{Technical estimates} 
We refer the reader  to \cite{MR2465733} for more details in the semiclassical analysis on the Fock space.  Here, we only sketch some useful technical results based in the above work. Remember that the KMS states $\omega_\varepsilon$, given by \eqref{sec3.eq.1},  are normal and so we denote,
$$
\omega_\varepsilon(\cdot)=\tr_\mathfrak{F}\left(\rho_\varepsilon\;\cdot\right)\,. 
$$
Furthermore, assume  for a subsequence $(\varepsilon_k)_{k\in \N}$, such that $\lim_{k\to\infty}\varepsilon_k= 0$,  that the set $\{\rho_{\varepsilon_k}\}_{k\in\N}$ admits a unique Wigner measure $\mu$. Then the following result holds true. 
\begin{lemma}
\label{apx.lem-est2} 
For any $\chi \in\mathscr{C}_0^\infty(\R)$ and $f,g\in\ell^2(G)$,  the set $\{ W(f) \rho_{\varepsilon_k} \chi(N_{\varepsilon_k}) W(g)\}_{k\in\N}$ admits a unique Wigner measure given by
$$
\big\{\mu \,e^{\sqrt{2}i \Ree\langle f+g,u\rangle} \chi (\|u\|^2)\big\}\,.
$$
\end{lemma}
\begin{proof}
For simplicity, we denote $\varepsilon$ instead of $\varepsilon_k$.  It is enough to prove that the set of Wigner measures for the density matrices $\{\rho_{\varepsilon} \chi(N_{\varepsilon}) \}$ is the singleton
$$
\{\mu \;\chi(\|u\|^2)\}\,. 
$$
In fact, using the Weyl commutation relations, one checks  according to \eqref{sec4.eq.wigner},
\begin{equation*}
\lim_{\varepsilon\to 0}  \tr_\mathfrak{F}\left( W(f) \rho_{\varepsilon} \chi(N_\varepsilon)  W(g) \;W(\eta) \right)= 
\int_{\ell^2(G)}  e^{i\sqrt{2} \Ree\langle f+g+\eta,u\rangle} \;d\nu\,,
\end{equation*}
where $\nu$ is a Wigner measure of the set of density matrices $\{\rho_{\varepsilon} \chi(N_{\varepsilon}) \}$.  Now, using  Pseudo-differential calculus,
$$
\chi(N_\varepsilon)=\left(\chi(\|u\|^2)\right)^{Weyl}+ {O}(\varepsilon)\,,
$$
where the subscript refers to the Weyl $\varepsilon$-quantization and  the difference between the right and left operators is of order $\varepsilon$ in norm  (see e.g. \cite[Thm.~8.7]{MR1735654}). Then \cite[Thm.~6.13]{MR2465733} with  Lemma \ref{particle_number}, gives 
$$
\nu=\mu\; \chi(\|u\|^2)\,.
$$
\end{proof}

\begin{lemma}
\label{apx.lem-est1} 
For any $\chi \in\mathscr{C}_0^\infty(\R)$ and $f\in\ell^2(G)$, there exists $c>0$ such that for all $\varepsilon\in(0,\bar\varepsilon)$,
\begin{equation*}
\|  \chi(N_\varepsilon)\, [N_\varepsilon, W(f)]\; \chi(N_\varepsilon) \|\leq c\; \varepsilon\,, \qquad \text{ and } \qquad 
\|  \chi(N_\varepsilon)\, [H_\varepsilon, W(f)]\; \chi(N_\varepsilon) \|\leq c\; \varepsilon\,.
\end{equation*}
\end{lemma}
\begin{proof}
The proof of the two inequalities are similar. We sketch the second one. 
Using standard computation in the Fock space (see e.g.~\cite[Proposition 2.10]{MR2465733}),
\begin{eqnarray*}
 [H_\varepsilon, W(f)]&=& W(f) \; \left( h(\cdot+i\frac{\varepsilon}{\sqrt{2}} f)-h(\cdot)\right)^{Wick}\,,
\end{eqnarray*}
where the subscript refers to the \emph{Wick} quantization, see \cite[Section 2]{MR2465733}, and $h$ is the classical Hamiltonian in \eqref{hamdnls}. By Taylor expansion, one writes 
\begin{equation*}
h(u+i\frac{\varepsilon}{\sqrt{2}} f)-h(u)=\varepsilon\,  C_\varepsilon(u)\,,
\end{equation*}
where $C_\varepsilon(u)$ is a polynomial in $u$ which can be computed explicitly. Using the number estimate in 
\cite[Lemma 2.5]{MR2465733}, one proves the inequality.  
\end{proof}

\section*{Acknowledgements}
The authors are grateful to Jean-Bernard Bru, Sylvain Gol\'enia and Vedran Sohinger for  helpful discussions. 

\bibliographystyle{plain}

\end{document}